\documentclass[10pt]{article}
\usepackage{arxiv}

\usepackage{graphicx}
\usepackage{float}
\usepackage{xcolor}
\usepackage{booktabs}
\usepackage{amsmath} %
\usepackage{amssymb}  %
\usepackage{amsthm}
\usepackage[title]{appendix}
\usepackage[style=ieee,backend=bibtex]{biblatex}
\title{\LARGE \bf
Data-Driven Reachable Set Computation using Adaptive Gaussian Process Classification and Monte Carlo Methods  %
}

\author{ \parbox{3 in}{\centering Alex Devonport\\
         Electrical Engineering and Computer Sciences\\
         University of California, Berkeley\\
         {\tt\small alex\_devonport@berkeley.edu}}
         \parbox{3 in}{ \centering Murat Arcak\\
         Electrical Engineering and Computer Sciences\\
         University of California, Berkeley\\
         {\tt\small arcak@berkeley.edu}}
} 

\newcommand{\R}{\ensuremath{\mathbb{R}}}
\newcommand{\initialset}{\ensuremath{\mathcal{X}_0} }
\newcommand{\finalset}{\ensuremath{\mathcal{X}_1} }
\newcommand{\eventset}{\ensuremath{\mathcal{X}_e} }
\newcommand{\inputset}{\ensuremath{\mathcal{U}} }

\newcommand{\appendixref}[1]{#1}

\newtheorem{theorem}{Theorem}

\begin{document}
\renewcommand*{\bibfont}{\small}

\maketitle
\thispagestyle{empty}
\pagestyle{empty}

\begin{abstract}

We present two data-driven methods for 
estimating reachable sets
with probabilistic guarantees.
Both methods make use of a probabilistic formulation allowing for a formal definition of a data-driven reachable set approximation that is correct in a probabilistic sense.
The first method recasts the reachability problem as a binary classification problem, using a Gaussian process classifier to represent the reachable set.
The quantified uncertainty of the Gaussian process model allows for an adaptive approach to the selection of new sample points.
The second method uses a Monte Carlo sampling approach to compute an interval-based approximation of the reachable set.
This method comes with a guarantee of probabilistic correctness, and an explicit bound on the number of sample points needed to achieve a desired accuracy and confidence.
Each method is illustrated with a numerical example.

\end{abstract}

\section{Introduction}

Reachable sets characterize the states to which a system may evolve using the knowledge of where it starts, what inputs may affect the system, and how long the system may evolve.
Computing reachable sets is a critical step in the solution to control problems involving objectives such as safety, recurrence, and more complicated requirements expressed as automata or temporal logic specifications.
However, accurate reachable sets are generally very expensive to compute, 
and common practice is to use a tractable relaxation, such as an overapproximation that is guaranteed to contain the true reachable set.

In relaxing the problem, the analyst must make a trade-off between computational tractability and accuracy of the overapproximation.
There are many reachable set overapproximation methods that lie at different points of the tractability-accuracy spectrum.
At one extreme, reachability methods based on the Hamilton-Jacobi-Bellman equations 
\cite{MitchellTomlinHJProjections, TomlinHJGames}
and dynamic programming \cite{BertsekasMinimax}, such as those used in the Level Set Toolbox \cite{mitchell2005toolbox}, yield reachable set approximations that are very accurate but slow to compute.
Zonotope-based methods
\cite{althoff2008verification}, such as those used in the CORA toolbox\cite{althoff2015introduction}, are faster to compute at the cost of some accuracy.
At the opposite extreme, interval reachability methods
\cite{meyer2019tira, meyer2017hierarchical, moor2002abstraction}
give overapproximations that require a minimum of resources to compute and store, but due to their strict geometry they are generally conservative.

In this paper we introduce a \emph{data-driven} approach that allows for improvements in both tractability and accuracy, at the cost of a relaxed guarantee of correctness.
The essence of this relaxation is to place a suitable probability measure over the initial set and the controls, and to define reachable sets as events on the induced probability space. Then, a sample of simulated system trajectories can be used to make probabilistic estimates of the true reachable set.
To achieve the lowest computational complexity possible, we minimize the number of sample trajectories while maintaining a probabilistic guarantee of a given accuracy.

Probabilistic methods have been used to analyze the reachability of stochastic systems
\cite{SastryStochasticHybridSystems, althoff2008stochastic, margellos2014road, yang2016multi} and as an exploratory tool to guide deterministic reachability analysis \cite{LygerosATCMC}.
Here, we investigate the probabilistic approach as a rigorous method in its own right to analyze the reachability of deterministic systems. 
Data-driven methods have also been used as a tool for robustness analysis of uncertain control systems
\cite{tempo2012randomized}, which allow for probabilistic verification of robustness against various types of uncertainty.
This paper provides a similar approach to the problem of reachable set computation.

We present two data-driven methods for computing reachable set approximations that make use of the probabilistic relaxation.
The first method uses a Gaussian process classifier (GPC) to construct a probabilistic reachable set of arbitrary accuracy.
The prediction uncertainty of the GPC allows us to employ an \emph{active learning} method
\cite{settles2009active}, where we sequentially select samples in order to maximize information gain.
The second method uses a Monte Carlo sampling approach to construct an interval overapproximation of the probabilistic reachable set.
Although less accurate, this method comes with a provable probabilistic guarantee. The two methods are complementary: the GPC method allows for approximations of 
higher accuracy (since it is not restricted to interval approximations),
while the Monte Carlo method can make faster approximations.
When probabilistic guarantees are acceptable for the problem at hand, the formalism and methods described in this paper can offer a significant computational speedup.
An additional advantage of the data-driven approach is that it may be used in a model-free way: we need only to be able to sample system trajectories, so the system itself is allowed to be a black box or otherwise inaccessible.
Indeed, many high-fidelity models are either available only in black-box form, or are too complex to analyze with standard reachability tools.

\section{Reachable Sets}

Suppose we have a dynamical system with state transition function $\Phi(t; t_0, x_0, u)$ that maps an initial state $x_0\in\R^n$ at time $t_0$ to a unique final state at time $t_1$, under the influence of an input $u\in C^{[t_0,t_1]}$ and the system dynamics.
For example, if the system is defined as a vector ordinary differential equation
\begin{equation}
\dot{x}(t)=f(x(t),u(t),t) \label{eq:vector_ode}
\end{equation}
whose solutions are well-defined and unique on the interval $[t_0,t]$,
then $\Phi(t; t_0, x_0, u)$ is the solution to (\ref{eq:vector_ode}) satisfying the initial condition $\Phi(t_0; t_0, x_0, u)=x_0$.

Now, suppose we have an \emph{initial set} $\initialset\subset\R^n$, and an \emph{input set} $\inputset\subset C^{[t_0,t_1]}$.
We would like to know all of the states to which the system may evolve between times $[t_0, t_1]$ starting in the initial set, and subjected to any allowable input.
The set of all such states is the \emph{forward reachable set}, which we write as
\begin{equation}
R_{[t_0,t_1]}=\{x | x=\Phi(t_1;t_0, x_0, u)\text{ for some } x_0\in\initialset, u\in\inputset\}.
\end{equation}
When the state transition function is invertible, we also consider the inverse of this problem.
Suppose we have a \emph{final set} set $\finalset\subset\R^n$, and we would like to know all of the states that can reach \finalset in the time $[t_0,t_1]$.
The set of all such states is called the \emph{backward reachable set}, that is
\begin{equation}
B_{[t_0,t_1]}=\{x | \Phi(t_1;t_0, x, u)\in\finalset\text{ for some } u\in\inputset\}.
\end{equation}
We may also be interested in finding the set \eventset of all initial states for which some \emph{event}, characterized by $h(x, t, u)=0$, occurs at some time $t_e\ge t_0$.
The set of all such states is called the \emph{event set}, that is
\begin{equation}
E_{t_0}=\{x | h(\Phi(t_e;t_0, x, u), t_e, u(t_e))=0 \mbox{ for some }t_e \ge t_0, u\in\inputset\}.
\end{equation}
This is similar to the backwards reachable set problem, except that $t_e$ is not known \emph{a priori}.
Further, $t_e$ will in general not be the same for each state that leads to the event.

\section{Probabilistic Reachable Sets}

To frame the data-driven approach, we consider a \emph{probabilistic relaxation} of the reachable set problems described above.
The methods in this paper consist of sampling initial states and inputs, evaluating the transition function at these sample points, and using the results to estimate the reachable set.
The state transition function may be available directly through numerical integration of (\ref{eq:vector_ode}), through more advanced computer simulations, or even through physical experiments.

A reachable set computed using a sample-based method can be at best only \emph{probabilistically accurate}, so we would like a way to represent this notion as well.
To formalize the notion of sampling from \initialset, we define a random variable $X_0 \sim p_0$ over the initial set.
The probability distribution $p_0: \initialset\to [0,1]$ is called the \emph{initial distribution}, and may be any distribution whose support is \initialset.
Similarly, we will define a random variable $U\sim p_u$ over the input set, with \emph{input distribution} $p_u:\inputset\to[0,1]$.

These two random variables, together with the state transition function, define the family of \emph{successor random variables} $X_t=\Phi(t;t_0, X_0, U)\sim p_t$ for $t\ge t_0$.
In general, the distribution $p_t$ will be unknown, since the state transition function is not known.
The successor distribution can be used to define a probability space whose sample space is the state space $\R^n$, whose events are the Borel sets of $\R^n$, and whose probability measure is $p_t$.
In this probability space, the probability of an event $\omega$ corresponds to the probability that the successor of a random initial state and input is an element of $\omega$.
This means that the true forward reachable set $R_{[t_0,t_1]}$ corresponds to the smallest event of probability 1.
With that in mind, we define the \emph{$\epsilon$-accurate reachable sets}, denoted $R_{[t_0,t_1],\epsilon}$ as the smallest events with probability $1-\epsilon$.
A set $R\subset\R^n$ such that $p_t(R)\ge 1-\epsilon$ is an \emph{overapproximation} of an $\epsilon$-accurate reachable set, since it must contain an $\epsilon$-accurate reachable set.
The relationship between the deterministic and probabilistic cases for forward reachable sets for a two-state system is illustrated in Figure \ref{fig:reachability-cartoon}.

\begin{figure}
\begin{center}
\includegraphics[width=0.45\textwidth]{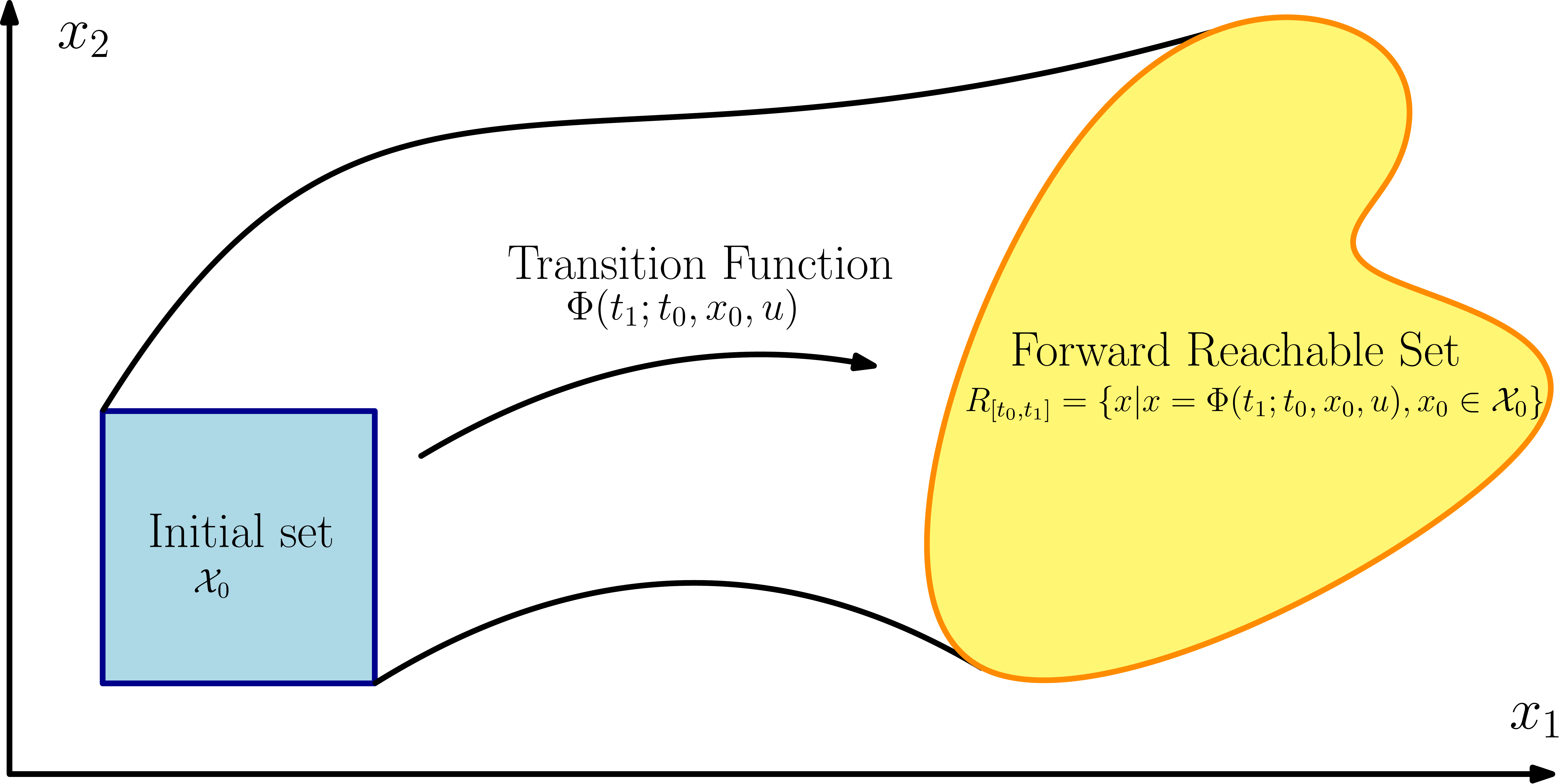} 
\includegraphics[width=0.45\textwidth]{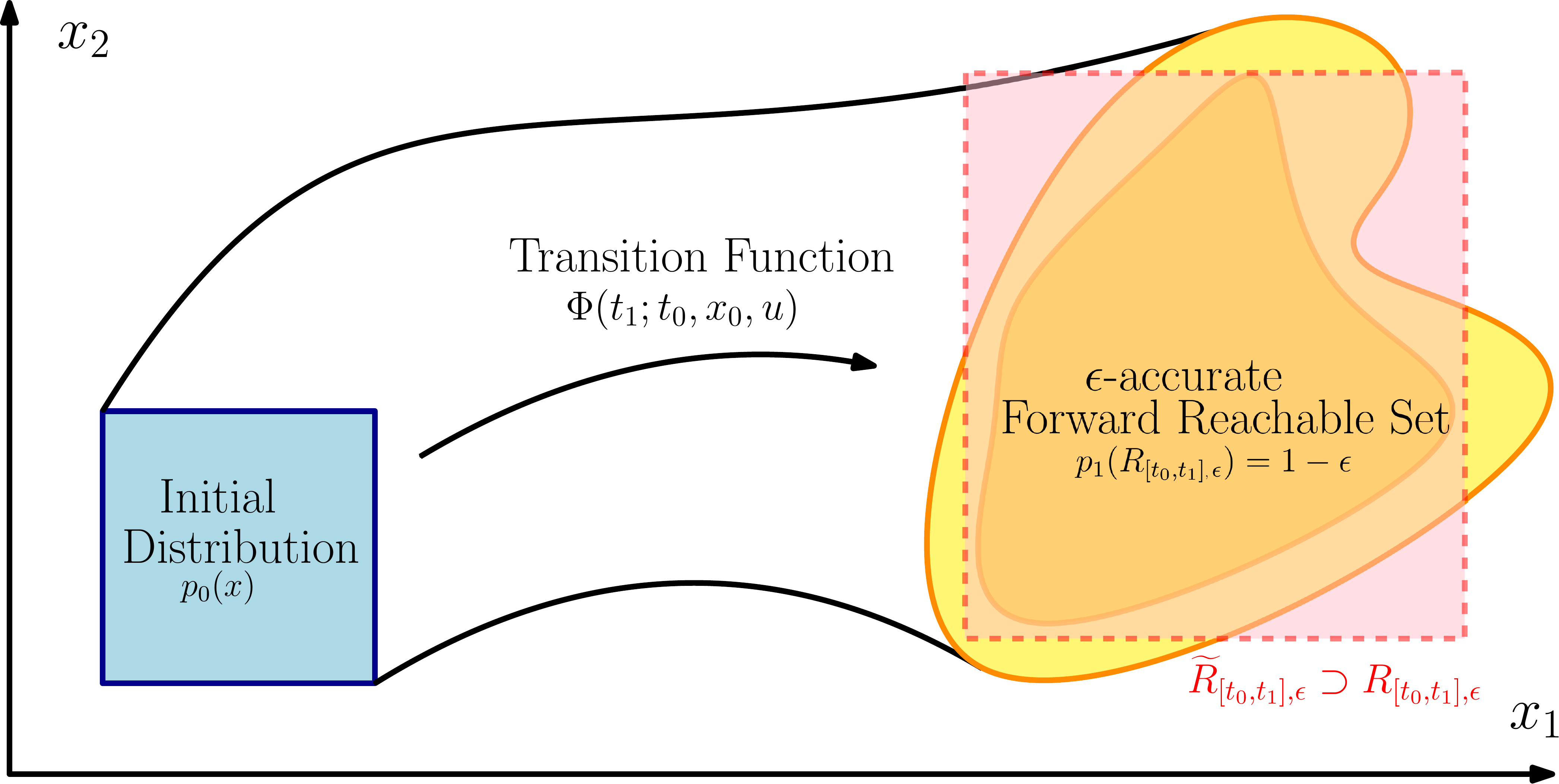}
\caption{A diagram of a forward reachable set (upper graph), and its relaxation to an $\epsilon$-accurate probabilistic forward reachable set (lower graph), and an overapproximation of the probabilistic reachable set. }
\label{fig:reachability-cartoon}
\end{center}
\end{figure}

We define a similar probabilistic formulation for backward reachable sets.
The only difference is that we will choose a final random variable $X_1$ and $U$, and let $X_0=\Phi^{-1}(t;t_0, X_1, U)$, where
\begin{equation}
\Phi^{-1}(t;t_0,x_1,u)=\{x | x_1=\Phi(t;t_0,x,u),u\in\inputset\}.
\end{equation}

For event sets, we are not interested specifically in the probabilistic behavior of $\Phi$, but instead in the likelihood that a given sample in an initial set will lead to the event.
Essentially, we would like to use samples to inform our belief about the location of the event set, so it is sensible to adopt a Bayesian formulation for the probabilistic event set.
We employ a distribution over the initial set, $p_0(x)$, which represents our belief that $x$ is in the event set prior to seeing any samples.
Then the posterior distribution conditioned on the sample trajectories represents an updated belief that the point $x$ belongs to the event set that takes information from the sample trajectories into account.
We call this posterior distribution the \emph{event distribution}, $p_e$.

\section{Gaussian Process Classification (GPC) With Adaptive Sampling}

The first method we present uses a Gaussian stochastic process to construct a binary classifier that identifies an estimate of the reachable set. 
A point in the state space is either in the reachable set or out of it, so determining the set of points in the reachable set has a natural representation as a binary classification problem.

A Gaussian process $g$ is a random variable defined over a space of functions with the property that the joint distribution of any finite selection of point evaluations of the function is distributed as a joint Gaussian random variable
\cite{gpml}.%
The covariance between any two point evaluations $g(x_1)$ and $g(x_2)$ is $k(x_1,x_2)$, where $k$ is the \emph{kernel function} of the process.

Suppose we have a set of $m$ sample points $x^{(i)}\in\R^n$ and their associated labels $y^{(i)}\in\{a,b\}$, where $a,b\in\R$, such that $y^{(i)}=a$ if $x^{(i)}$ is in the reachable set, and $y^{(i)}=b$ otherwise.
For example, a suitable choice of labels would be $a=1, b=0$.
We use a Gaussian process to construct a classifier that minimizes the \emph{regularized least-squares classification risk}, that is a function $g:\R^n\to\R$ that minimizes
\begin{equation}
\sum_{i=1}^{m} \left(g(x^{(i)})-y^{(i)}\right)^2 + ||g||_{k}
\end{equation}
where $||\cdot||_{k}$ is a norm that depends on the kernel function.
A classifier that minimizes this risk is called a \emph{least-squares classifier}. 
Least-squares classification is attractive here because the mean $\mu_{\hat{g}}$ and variance $\sigma_{\hat{g}}$ of the Gaussian process $\hat{g}$ that minimizes this risk have analytic expressions that can be computed quickly.

To make predictions using this classifier, we select a \emph{threshold} $\gamma\in (a,b)$, and declare that a point $x$ is predicted to be in the reachable set if $\hat{g}(x)\ge\gamma$, and not in the reachable set otherwise. For example, in the $a=1,b=0$ case, $\gamma=0.5$ is suitable.
With a threshold chosen, the reachable set estimate produced by this method is the sublevel set
\begin{equation}
\hat{R}=\{ x | \mu_{\hat{g}}(x) < \gamma\}.
\end{equation}
To construct a data set, we select a set of sample points $x^{(i)}$, and use the state transition function to assign a label $y^{(i)}$ to each of the sample points based on whether or not it is in the reachable set.

In principle, we may select the sample points in any way we like, e.g. uniform sampling over the region of interest, or using Latin hypercube sampling.
However, since we wish to minimize the number of transition function evaluations, we use the GPC model of the reachable set to inform our choice of future sample points. 
This kind of sampling is called \emph{adaptive sampling}, since our selection method adapts according to the incoming data, and is an \emph{active learning} method. 
The use of adaptive sampling to guide the construction of a Gaussian process model is motivated by a method from optimal experiment design known as \emph{Adaptive Kriging}
\cite{ak-mcs, schobi-pck}, in which a Gaussian process regression is used to form a \emph{surrogate model} for an expensive computational model.

To have the GPC inform our selection of sample points, we use the prediction uncertainty of the GPC to sequentially select sample points that maximize some measure of information gain.
The most relevant measure of information gain for this problem is the \emph{probability of misclassification}
\cite{bect2012sequential}, that is the probability that a state in the reachable set is identified as being outside of it, or vice versa.
With classifier threshold $\gamma$, the probability of misclassification is
\begin{equation}
P_{misclass}(x) = \Phi\left( -\frac{|\mu_{\hat{g}}(x) - \gamma|}{\sigma_{\hat{g}}(x)} \right)
\label{eq:p_misclass}
\end{equation}
where $\Phi$ in (\ref{eq:p_misclass}) is the cumulative distribution function of the standard normal distribution.

When selecting a new sample point, ideally we would like to find the point in the state space with the highest probability of misclassification.
However, this is a nonconvex and potentially high-dimensional optimization problem.
Instead of searching the entire state space for a new sample, we use a stochastic optimization approach proposed in
\cite{ak-mcs} and search over a large pool of randomly-selected \emph{candidate samples}. We calculate the probability of misclassification for each candidate, and select the one with the highest probability of misclassification to be the next sample.
The sample pool is selected using a Latin Hypercube, so that the candidate samples will be evenly distributed over a compact region of the state space.
Note that this is distinct from selecting samples directly by a Latin hypercube: after we have selected the candidate pool, only a small number of candidate points will be selected as sample points, and the distribution of the selected points will be guided by the probability of misclassification.

\subsection*{Example: Safe Set Estimation for Adaptive Cruise Control}

\begin{figure}
\includegraphics[width=\textwidth]{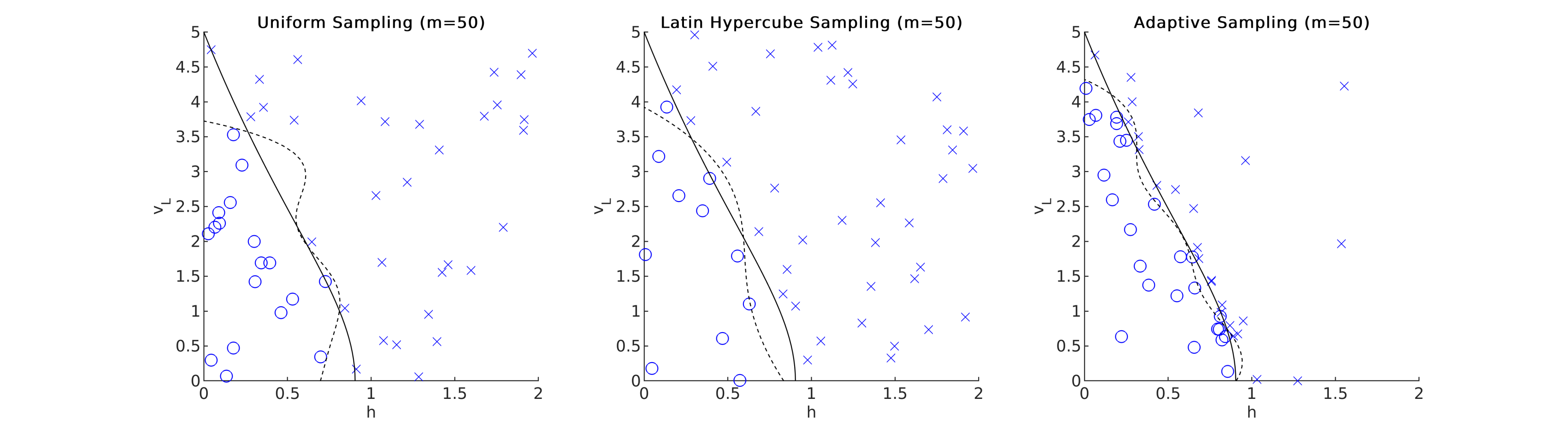}
\includegraphics[width=\textwidth]{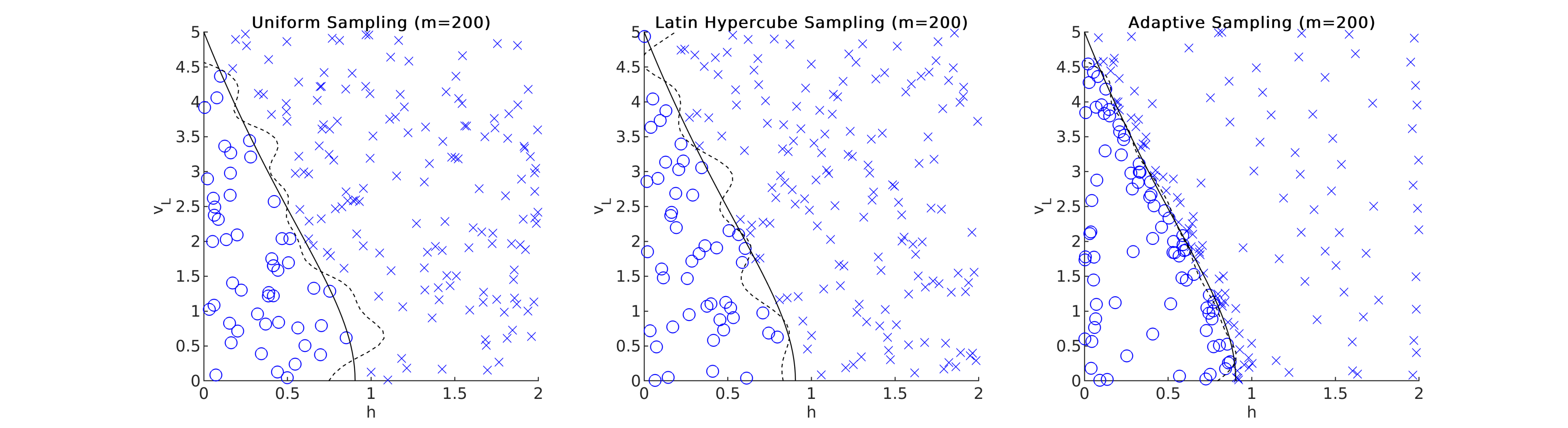}
\caption{Estimated safe set boundaries (dashed lines) for the ACC model computed with the GPC method compared with the true safe set boundary (solid lines), which is calculated analytically.
The sample locations are also shown: an `o' indicates that a collision occurred, and an `x' indicates that it did not.
The model parameters are set at $a=4.9$, $b=1$, and the initial follower velocity is fixed at $v_F(0)=5$.
Top row: $m=50$ sample points.
Bottom row: $m=200$ sample points.
For adaptive sampling, a candidate pool of $m_{candidate}=1000$ samples was used in both cases.
For the two sample sizes shown, adaptive sampling is able to make the most accurate approximation of the event set out of the three sampling methods used.}
\label{fig:safeset_gpc}
\end{figure}
Consider the Adaptive Cruise Control (ACC) scenario depicted in Figure \ref{fig:acc_diagram}.
In this scenario, a car being operated by ACC (the \emph{follower}) is driving behind another car (the \emph{leader}).
The follower and leader are initially traveling with positive velocities $v_F(0)$ and $v_L(0)$ respectively.
At $t=0$, the leader begins to brake and eventually comes to a halt.
If the distance between the leader and follower becomes zero at any $t>0$, then the two cars have collided.
To prevent this, we determine what initial states (that is, velocities and relative positions at $t=0$) give the follower enough time to prevent a collision.
We call the set of all such initial states a ``safe set''.

We use the following point-mass model for the dynamics of the two vehicles:
\begin{align}
\dot{h}(t) &= v_{L}(t)-v_{F}(t) \\
\dot{v}_{L}(t) &= -a - b v_{L}(t)^2\\
\dot{v}_{F}(t) &= -a - b v_{F}(t)^2 
\end{align}
where $v_L(t)$ and $v_F(t)$ are the velocities of the leader and follower, respectively, and $h(t)$ is the distance between the two cars.
The acceleration of each car has a constant term from the brakes, as both cars are applying the brakes fully, and a quadratic term from drag force.

\begin{figure}
\begin{center}
\includegraphics[width=0.45\textwidth]{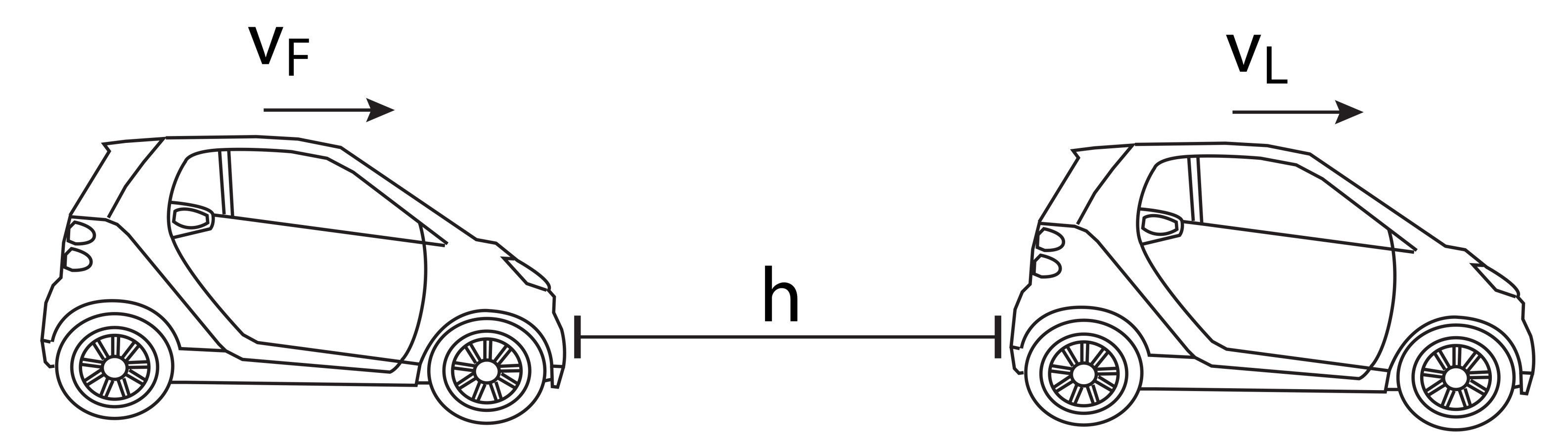}
\caption{Diagram of the leader and follower, and the associated state variables, in the ACC braking model. If $h(t_e)=0$ for some $t_e\ge t_0$, the two cars have collided.}
\label{fig:acc_diagram}
\end{center}
\end{figure}

This problem is an event set estimation problem because the safe set we wish to determine corresponds to the complement of the set of initial conditions $x(0)= \begin{bmatrix} h(0) & v_{L}(0) &  v_{F}(0) \end{bmatrix}^T$ for which the event $h(x)=h=0$ occurs.

\appendixref{
This model can be solved analytically, and the true event set is shown in the Appendix to be 

\begin{equation}
 E = \{ (h,v_L,v_F)| 
   h 
   + \frac{1}{2b}\log\left(1 + \frac{b}{a}v_L^2 \right)
   - \frac{1}{2b}\log\left(1 + \frac{b}{a}v_F^2 \right)
   \ge 0
\}.
\end{equation}
Since we know the true event set, we can directly observe how well the GPC method approximates the true event set under different conditions. 
}

For convenience of visualization, we hold the initial velocity of the follower constant at $v_{F}(0)=0.5$.
We restrict our attention to a compact region of the state space, specifically
\begin{align}
 0 \le h &\le 2\\
 0 \le v_{L} &\le 5.
\end{align}

Using sample points from this region, we construct a least-squares GPC using a \emph{squared-exponential} kernel, that is we take
\begin{equation}
 k(x_1,x_2)=\sigma \exp(-(x_2-x_1)^T \Lambda (x_2-x_1)),
\end{equation}
where $\sigma$ and the diagonal matrix $\Lambda=\text{diag}(\ell_1,\ell_2,\ell_3)$ are \emph{hyperparameters} that are selected using maximum likelihood.

To demonstrate the effectiveness of the adaptive sampling method, we compare it to two other non-adaptive sampling strategies: sampling uniformly at random over the region, and sampling with a Latin hypercube over the region.
To demonstrate how the number samples affects the quality of the predicted event set, we form two sample sets for each of the sampling methods, with $m=50$ and $m=200$ each.

For adaptive sample selection, we begin by selecting a pool of $m_{candidate}=1000$ candidate samples from the region of interest using a Latin hypercube.
Three samples are selected at random to serve as the initial set for the GPC model, and the remaining $m-3$ are selected by sequentially minimizing the probability of misclassification.

The GPC-estimated event sets are shown in Figure \ref{fig:safeset_gpc}.
The true reachable set is also shown, to confirm that the estimated reachable sets are converging to the ground truth.
For both sample sizes, the adaptive sampling method makes the most accurate event set estimate out of each of the three methods.
By maximizing the probability of misclassification with each new sample, the adaptive method will either select a new sample with high prediction variance, which will be far away from the other samples, or one whose prediction mean is close to the threshold; that is, one close to the border.
After minimizing the variance over most of the state space region of interest, the adaptive method begins to select samples that are likely to be near the border of the event set.
Using this strategy, the adaptive GPC reduces the number of samples required to identify points near the boundary of the event set.
By contrast, the uniform and Latin hypercube methods can only select states at random, so the likelihood of selecting a sample point near the boundary of the event set never increases.

\appendixref{
We used a point-mass model here so that the true event set could be derived, as has been done in the Appendix. However, this same analysis could be carried out with a high-fidelity model with no change to the methodology.
}

\section{Monte Carlo Interval Overapproximation}

We now present a Monte Carlo Sampling (MCS) approach to produce \emph{interval} overapproximations of epsilon-accurate reachable sets, that is overapproximations of the form
\begin{equation}
\hat{R} = [\underline{x}, \overline{x}]= \{x | \underline{x} \le x \le \overline{x}, \underline{x}\in\R^n,\overline{x}\in\R^n\}
\end{equation}
where $\le$ is the vector inequality corresponding to the positive orthant cone of $\R^n$.
Geometrically, the set $[\underline{x}, \overline{x}]$ is an axis-aligned hyperrectangle of dimension $n$ whose least point is $\underline{x}$ and whose greatest point is $\overline{x}$.

In general, even the tightest interval overapproximation of the reachable set will be inaccurate, since intervals have such restricted geometry.
Despite the inaccurate nature of their approximation of the true reachable set, interval approximations can be an appropriate design choice when quick computation and low memory requirements are preferable to reachable set accuracy.

An important example of when interval approximation is a suitable design choice is \emph{symbolic control}, where controller synthesis is carried out on a finite-state machine \emph{abstraction} that simulates the continuous-state dynamical system
\cite{lunze1994qualitative,
alur2000discrete,
moor2002abstraction,
tabuada,
gazit11,
belta, 
7519063}%
. The states of the abstraction represent the cells of a partition of $\R^n$, and the transitions are derived from the intersection of the forward reachable sets of each cell with the other cells. 
For high-dimensional state spaces, the number of reachable sets that must be computed and stored grows rapidly, so it is necessary to use a reachable set overapproximation that can be computed quickly and is memory-efficient.

A simple method to calculate the interval approximation is a Monte Carlo approach.
For the forward reachable set case, this would consist of the following steps:
\begin{enumerate}
\item take a set of $m$ samples each from the initial distribution and input distribution, $\{x_0^{(i)}\}_{i=1}^{m}$ and $\{u^{(i)}\}_{i=1}^{m}$;
\item Evaluate the sample successor states $x_1^{(i)}=\Phi(t;t_0,x_0^i,u^i)$;
\item Take $\hat{R}^{(m)}$ as the smallest interval containing all of the $x_1^{(i)}$.
\end{enumerate}

Despite its simplicity, the Monte Carlo Sampling (MCS) method described above is provably effective at overapproximating $\epsilon$-accurate reachable sets with intervals. In particular, the inequality (\ref{eq:mc_sample_bound}), adapted from an example in \cite{Vidyasagar} serves as a lower bound on the number of sample points required to ensure that the method described above produces an overapproximation of a desired accuracy and confidence.
Although similar in form, the sample complexity bound (\ref{eq:mc_sample_bound}) is distinct from the Chernoff bounds commonly used in the analysis of Monte Carlo methods, in ways that will be discussed after the proof.

\begin{theorem}
Let $\epsilon$, $\delta\in (0,1)$.
If 
\begin{equation}
m \ge \frac{2n}{\epsilon}\log\left(\frac{2n}{\delta}\right),
\label{eq:mc_sample_bound}
\end{equation}
then $\hat{R}^{(m)}$ overapproximates an $\epsilon$-accurate reachable set with confidence $\delta$, i.e.
$P(R_{[t_0,t_1],\epsilon} \subset \hat{R}^{(m)}) \ge 1-\delta$.
\end{theorem}

\begin{proof}
First, suppose that we have $m$ samples of the successor $\{x_1^{(i)}\}_{i=1}^m$, and that the resulting $\hat{R}^{(m)}$ overapproximates the $\epsilon$-accurate reachable set.
 
An $n$-dimensional interval has $2n$ ``faces'', which are all axis-aligned.
To each ``face'' $f_i$ of the interval we associate the half space $\mathcal{H}_i$ whose boundary hyperplane coincides with that face, and faces ``away from'' the interval.

Now, let $\mathcal{P}_i$ be a half space parallel to $\mathcal{H}_i$ such that $p_{t_1}(\mathcal{P}_i)\ge\frac{\epsilon}{2n}$. Furthermore, let $\mathcal{P}_i$ be the \emph{smallest} such half space, in the sense that any other half space  $\mathcal{P}$ parallel to $\mathcal{H}_i$ such that $p_{t_1}(\mathcal{P}_i)\ge\frac{\epsilon}{2n}$ contains $\mathcal{P}_i$. By the right-continuity of $p_1$, this half space is unique. In the case that $p_1$ is continuous, this is just the unique hyperplane parallel to $\mathcal{H}_i$ such that $p_{t_1}(\mathcal{P}_i)=\frac{\epsilon}{2n}$.

The probability that none of the $m$ samples is in the half space $\mathcal{P}_i$ is $(1-\frac{\epsilon}{2n})^m$.
This implies that the probability that at least one of the $\mathcal{P}_i$ contains no samples less than $2n(1-\frac{\epsilon}{2n})^m$, and that the probability that each $\mathcal{P}_i$ contains at least one sample is greater than $1-2n(1-\frac{\epsilon}{2n})^m$.

If one of the samples is in $\mathcal{P}_i$, then $\mathcal{H}_i\subset\mathcal{P}_i$, meaning that $p_{t_1}(\mathcal{H}_i) \le p_{t_1}(\mathcal{P}_i)\le\frac{\epsilon}{2n}$. If this is true for each $\mathcal{P}_i$, then $p_{t_1}(\bigcup_i \mathcal{H}_i)=p_{t_1}(\hat{R}^{(n)\complement})\le\epsilon$, and so $p_{t_1}(\hat{R}^{(i)})\ge 1-\epsilon$. This implies that
\begin{equation}
P(p_{t_1}(\hat{R}^{(i)})\ge 1-\epsilon) \ge 1-2n(1-\frac{\epsilon}{2n})^m.
\end{equation}
To ensure that $p_{t_1}(\hat{R}^{(i)})\ge 1-\epsilon)$, i.e. that  $\hat{R}^{(m)}$ overapproximates contains an $\epsilon$-accurate reachable set, with probability $\ge 1-\delta$, it suffices to ensure
\begin{equation}
\delta \ge 2n\left(1-\frac{\epsilon}{2n}\right)^m.
\label{eq:mc_implicit_bound}
\end{equation}
Using the identity $\log(1-x) \le -x$, valid for $0\le x \le 1$, we can conclude that (\ref{eq:mc_sample_bound}) guarantees (\ref{eq:mc_implicit_bound}).
\end{proof}
For standard Monte Carlo methods, such as Monte Carlo Integration, probability inequalities such as \emph{Chernoff bounds} can be used to determine the number of samples needed for a desired accuracy. 
However, Chernoff bounds are not applicable in this case. 
Chernoff bounds apply to the tail probability of a sum of iid random variables (e.g. a Monte Carlo integral) exceeding an arbitrary parameter, 
and are derived by applying Markov's inequality to the moment generating function of the random sum.
In our case, however, we are interested in bounding the tail probability of the \emph{maximum} of an iid sequence of random variables exceeding a specific observed value (the empirical maximum). The strategy for deriving a Chernoff bound does not work in this case, which is why a different strategy was needed to prove Theorem 1.

\subsection*{Example: Robustness Analysis For a Powered Lower Limb Orthosis Through Forward Reachable Sets}

\begin{figure}
\begin{center}
\includegraphics[width=0.3\textwidth]{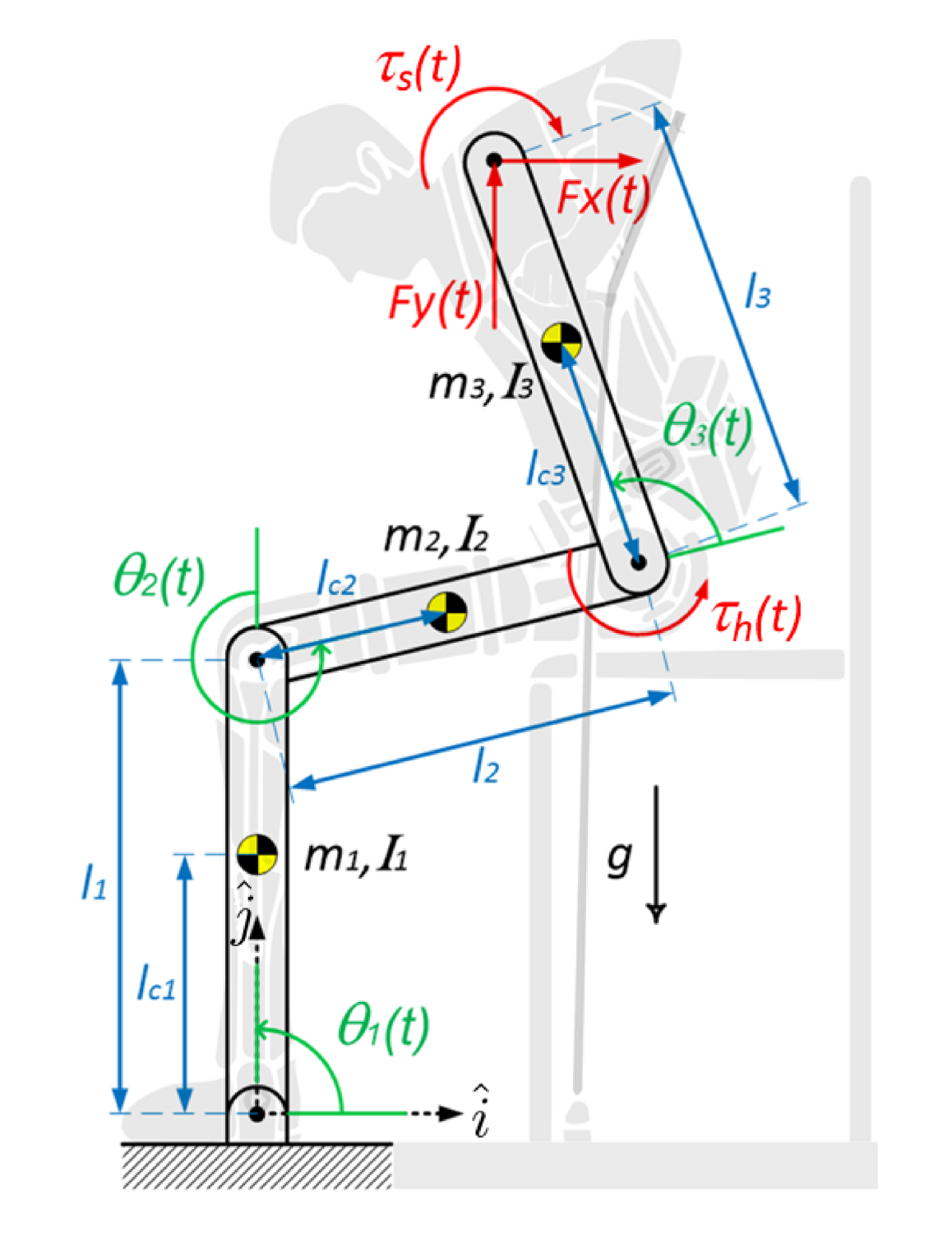}
\caption{Diagram of the orthosis and its user, from \cite{na}. State variables are labeled in red and green. Parameters are labeled in blue and black.}
\label{fig:ortho_diagram}
\end{center}
\end{figure}

We now present an application of the MCS method to a reachable-set problem posed in \cite{na} for robustness analysis of a sit-to-stand motion controller for powered lower limb orthoses.

The Powered Lower Limb Orthosis and its user are modeled in \cite{na} as a three-link planar robot with three joints.
The model has six states (angles and angular velocities of the three joints, labeled in red and green in Figure \ref{fig:ortho_diagram}) and twelve parameters (the lengths, masses, and moments of inertia of the three links, and the distances between each joint and the center of moment of its corresponding link, labeled in blue and black in Figure \ref{fig:ortho_diagram}).
The twelve parameters are \emph{uncertain}, since they all depend on the weight of the user.
The sit-to-stand reference trajectory is defined in terms of the position in the $x-$ and $y-$ (horizontal and vertical in Figure \ref{fig:ortho_diagram}) directions of the center of mass (CoM) of the user and orthosis, which can be computed from the state variables and parameters.

The authors of \cite{na} designed a finite time horizon LQR controller to track a reference trajectory that brings the user and orthosis from a sitting position to a standing position.
To analyze the robustness of the motion to parameter variations, the authors used sensitivity-based reachability methods to compute an interval overapproximation of the CoM trajectories subject under parameter changes induced by a 5\% variation in user body weight.
In this example, we approach this robustness analysis problem by using the MCS method
to compute an interval overapproximation of the CoM trajectory with accuracy $\epsilon=0.05$ and confidence $\delta=0.001$.

To recast the problem from a robustness analysis problem into a reachability problem, we add a new state variable $p_i$ for each parameter with constant dynamics, i.e. $\dot{p_i}=0$.
The parameters may be selected by choosing the initial conditions of the $p_i$.
This way, we can perform the robustness verification by computing reachable set overapproximations of the parameter-augmented 18-state system.

The parameter intervals are derived from human biometric data, and the initial state is assumed to be fixed. We choose an initial distribution to be the uniform distribution over the initial set. In this case, the initial distribution has a physical significance, namely that the parameters are subject to change from user to user. If we had more detailed statistical information about the parameters (e.g. that they are Gaussian-distributed with some mean and variance), that could be used to define the initial set instead.

Using (\ref{eq:mc_sample_bound}) with values $\epsilon=0.05$, $\delta=0.001$, and $n=18$, we know that $m=7554$ sample trajectories will suffice to compute an $\hat{R}^{(m)}$ that has at the desired levels of accuracy and confidence.

\begin{figure}
 \begin{center}
  \includegraphics[width=0.5\textwidth]{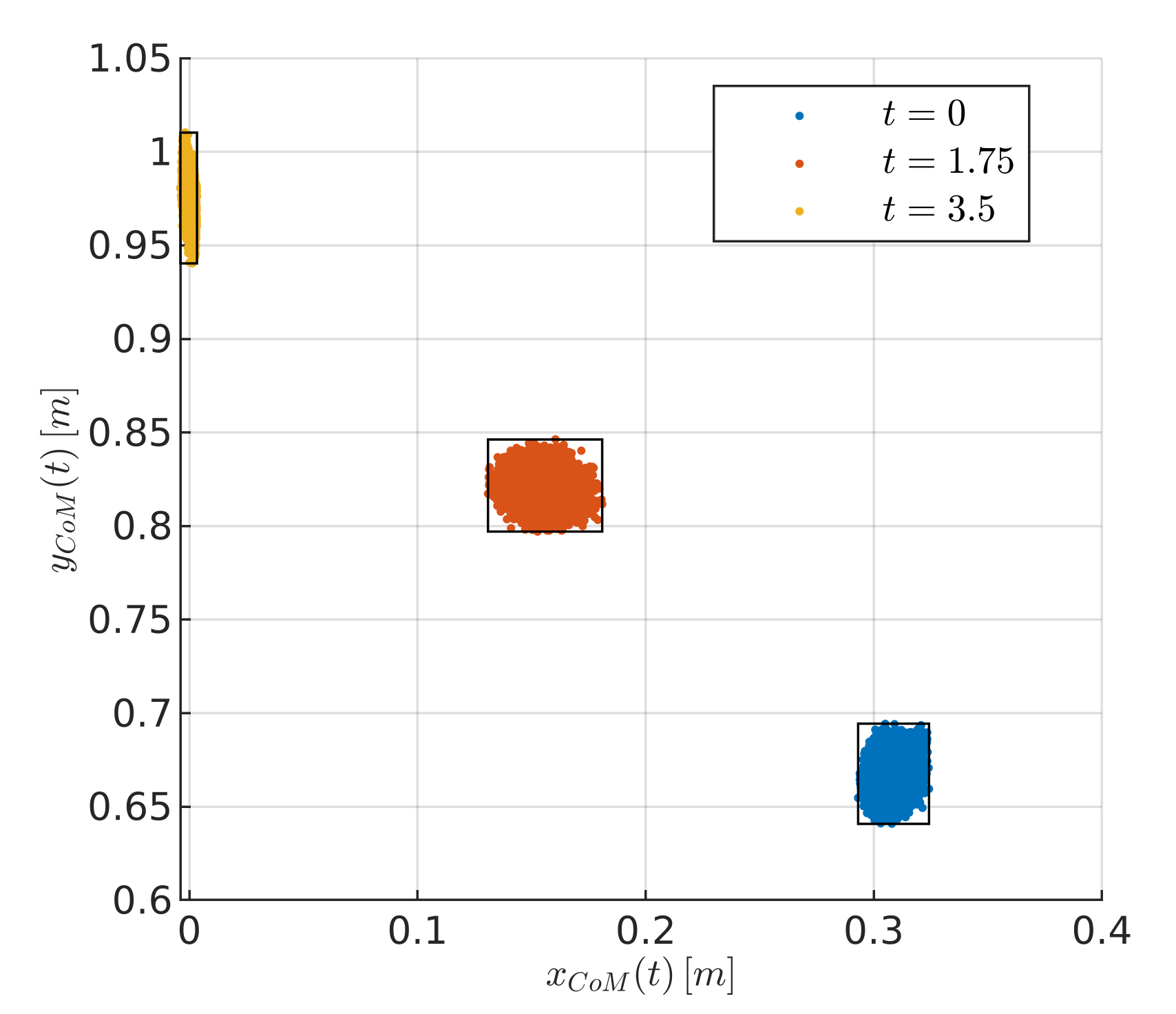}
  \caption{Reachable set overapproximations (black boxes) of the center of mass trajectories at three time instants over sit-to-stand movement, calculated using the Monte Carlo interval overapproximation method with $\epsilon=0.05$, $\delta=0.001$.
$m=7554$ samples trajectories (also shown) were required to ensure the specified accuracy and confidence over the parameter-augmented 18-dimensional state space.}
\label{fig:ortho_result}
 \end{center}
\end{figure}

The resulting interval overapproximations of the position of the CoM and is shown in Figure \ref{fig:ortho_result}. Specifically, we show the overapproximation for three points in the sit-to-stand movement; at the beginning ($t=0$), in the middle ($t=1.75$), and at the end ($t=3.5$) of the movement. In addition, the volumes of the intervals in Figure \ref{fig:ortho_result} are reported in Table \ref{tab:ortho_volumes}, and compared to the volumes of the intervals computed using sensitivity-based methods in \cite{na}. The table also reports the volumes of the interval overapproximations to the velocity of the CoM, which come from the same calculation as the position overapproximations. At each time, the volumes of the MCS rectangles is lower than those of the sensitivity method volumes, meaning that the MCS method gives less conservative estimates of the reachable sets.

\begin{table}
 \begin{center}
\begin{tabular}{lll} \toprule
interval & Volume from \cite{na}  & Volume from MCS \\
 \midrule
 $(x_{CoM}, y_{CoM})$, $t=0$ & $2.4\times 10^{-3}$ & $1.7\times 10^{-3}$\\
 $(x_{CoM}, y_{CoM})$, $t=1.75$ & $4.0\times 10^{-3}$ & $2.5\times 10^{-3}$\\
 $(x_{CoM}, y_{CoM})$, $t=3.5$ & $3.2\times 10^{-3}$ & $5.0\times 10^{-4}$\\
 $(\dot{x}_{CoM}, \dot{y}_{CoM})$, $t=0$ & $0$ & $0$\\
 $(\dot{x}_{CoM}, \dot{y}_{CoM})$, $t=1.75$ & $1.0\times 10^{-3}$ &$7.4\times 10^{-4}$\\ 
 $(\dot{x}_{CoM}, \dot{y}_{CoM})$, $t=3.5$ & $4.8\times 10^{-6}$ &$1.1\times 10^{-6}$\\ 
 \bottomrule
 \end{tabular}
 \end{center}
\caption{Comparison of Interval Overpproximation Volumes}
\label{tab:ortho_volumes}
\end{table}

\section{Conclusion}
In this paper, we showed the utility of a probabilistic formulation of the problem of reachable set approximation by developing data-driven reachable set computation methods, and by verifying them with numerical examples. Between the two methods, the analyst can make a reachable set approximation to any degree of accuracy. Sets computed with the MCS method also come with a rigorous probabilistic guarantee of accuracy and confidence. Since the problem formulation makes few assumptions, these two methods can be used on a wide range of systems, many of which are too complex for standard reachability methods. 

Both of the methods we present are useful in their present form, as shown by the examples. However, they can both be improved and extended with future work.
For the MCS interval overapproximation method,
the proof of (\ref{eq:mc_sample_bound}) makes no assumptions on the successor distribution $p_1$, other than that $p_1(x)$ be well-defined for all $x\in\initialset$.
On one hand, having a distribution-free result is useful, since we don't in general have any knowledge of $p_1$.
On the other hand, it suggests that if we did have some knowledge about $p_1$, the bound could perhaps be improved by using importance sampling or a related technique.
Additionally, the proof as given depends on the geometry of hyperrectangles, and does not generalize to other classes of sets that we may want to use for interval reachable set overapproximation.
For the GPC method, an useful extension would be a probabilistic guarantee, of a similar kind to the guarantee that the MCS method enjoys. It may also be useful to consider other classification methods than least-squares, which may be more accurate and amenable to analysis.

\section*{Acknowledgments}
\begin{flushleft}
This work was supported in part by the grants ONR N00014-18-1-2209, AFOSR FA9550-18-1-0253, NSF ECCS-1906164.
\end{flushleft}

\printbibliography

\section*{Appendix: Derivation of the True ACC Safe Set}

For the system
\begin{eqnarray}
\dot{x}&=&v, \\
\quad \dot{v}&=&-bv^2-a, \label{vdot}
\end{eqnarray}
the distance traveled before coming to a full stop is (from the derivation below):
\begin{equation}\label{main}
\frac{1}{2b}\ln\left(1+\frac{b}{a}v(0)^2\right).
\end{equation}

Since $h$ is a monotonic function of $t$ (see below), to check for a collision it suffices to check the sign of $h$ after both cars have come to a stop. if $h\ge0$ when both cars have stopped, then $h(t)\ge0$ for all prior $t$, and no collision occurred. On the other hand, if $h<0$ after the cars stopped, at some point $h$ changed sign, at which time there was a collision. With this in mind,
from (\ref{main}) the safe set in the $(h(0),v_L(0))$ space is:
\begin{multline}
E(v_F(0))=\bigg\{(h,v): h+\frac{1}{2b}\ln\left(1+\frac{b}{a}v^2\right)\\
-\frac{1}{2b}\ln\left(1+\frac{b}{a}v_F(0)^2\right)\ge 0\bigg\}.
\end{multline}

Derivation of (\ref{main}):
the solution of (\ref{vdot}) is
\begin{equation}\label{v}
v(t)=\sqrt{\frac{a}{b}}\tan\left( \tan^{-1}\left(\sqrt{\frac{b}{a}}v(0) \right)-\sqrt{ab}\,t\right)
\end{equation}
and integration gives
\begin{equation}\label{x}
x(t)-x(0)=\frac{1}{b}\ln \left( \frac{\cos\left(\tan^{-1}\left( \sqrt{\frac{b}{a}}v(0)\right)-\sqrt{ab}\,t\right)}{\cos\left(\tan^{-1}\left( \sqrt{\frac{b}{a}}v(0)\right)\right)}\right).
\end{equation}
Note  that $v(t)=0$ when $t$ is such that the argument of the tangent term in (\ref{v}) is zero.  Since the cosine in the numerator of (\ref{x}) has the same argument, the numerator equals one when $v(t)=0$.  When, in addition, the denominator is simplified as
$$
\cos\left(\tan^{-1}\left( \sqrt{\frac{b}{a}}v(0)\right)\right)=\left( 1+\frac{b}{a}v(0)^2\right)^{-1/2},
$$
the right-hand side of (\ref{x}) becomes (\ref{main}).

Monotonicity of $h(t)$:
introducing some abbreviations, we can write (\ref{x}) as
\begin{equation}
 x_i(t)-x_i(0) = \frac{1}{b} \ln \left(\frac{\cos(\alpha_i-\beta t)}{\cos(\alpha_i)}\right)
 \label{x simplified}
\end{equation}
where $i=L\text{ or }F$.
This This expression gives us $\dot{h}$ as
\begin{equation}
 \dot{h}=\dot{x}_L-\dot{x}_F=\frac{\beta}{b} \left(\tan(\alpha_L-\beta t)-\tan(\alpha_F-\beta t) \right).
\end{equation}
Now we can consider two cases: $\alpha_L>\alpha_F$, and $\alpha_L\le\alpha_F$. In the first case, we will have $\alpha_L-\beta t > \alpha_F-\beta t$, and $\tan(\alpha_L-\beta t) > \tan(\alpha_F-\beta t)$. This makes $\dot{h}(t)>0$ for all $t$, meaning that $h(t)$ is monotonically increasing. By analogous reasoning, in the second case $h(t)$ is monotonically decreasing.

\end{document}